\documentclass{article}
\usepackage[utf8]{inputenc}
\usepackage{amsmath}
\usepackage{amssymb}
\usepackage{graphicx}
\usepackage{geometry}
\usepackage{amsthm}
\usepackage{cite}

\newtheorem{theorem}{Theorem}
\newtheorem{lemma}{Lemma}
\newtheorem{corollary}{Corollary}
\newtheorem{remark}{Remark}
\newtheorem{definition}{Definition}

\newcommand{\davg}[2]{\bar{\bar{#1}}^{#2}}
\newcommand{\dcavg}[2]{\tilde{\tilde{#1}}^{#2}}

\title{Operator Analysis of MACD}
\author{Yuelong Li}
\date{March 2020}

\begin{document}
\setlength{\parindent}{0pt}
\setlength{\parskip}{1ex}

\maketitle
\begin{abstract}
This paper provides a rigorous functional-analytic treatment of the MACD (Moving Average Convergence Divergence) indicator \cite{appel1979, murphy1999}, a classical tool in technical analysis. We demonstrate that MACD, typically viewed as the difference of two moving averages, admits a precise interpretation as a \textit{phase-corrected, smoothed derivative operator} \cite{oppenheim2009, burrus1998}.

By analyzing nested and recursive moving averages, we reveal that MACD is structurally equivalent to a band-pass filter and derive exact formulas expressing MACD as a finite difference of delayed and doubly-averaged signals. We prove new operator identities showing that MACD corresponds to the derivative of a phase-centered, double-smoothed average, delayed to correct for asymmetry introduced by causal averages. 

These results unify MACD with tools from harmonic and functional analysis, and provide a principled basis for understanding its behavior in signal detection, filtering, and trend analysis. The framework naturally generalizes to recursive decompositions, culminating in an expansion that expresses MACD as a weighted sum of delayed, smoothed derivatives — revealing the true analytical structure of this widely used but often heuristically applied indicator.
\end{abstract}

\section*{Introduction}

The Moving Average Convergence Divergence (MACD) indicator is a widely used tool in technical analysis \cite{murphy1999, pring2002}, popular for detecting momentum shifts and trend reversals in time series data \cite{appel1979}. Traditionally defined as the difference between two moving averages of different window lengths, MACD has gained popularity for its empirical effectiveness \cite{fama1992} — yet its analytical structure remains underexplored.

In this paper, we revisit MACD not as a heuristic rule, but as a precise mathematical operator. By treating moving averages as integral and convolutional operators, we show that MACD can be interpreted as a derivative — specifically, a centered, smoothed, and delayed derivative of the signal. This interpretation sheds new light on its behavior, particularly in noisy or nonstationary signals.

The key contributions of this work are:

\begin{itemize}
  \item We define phase-corrected and double-averaged smoothing operators to rigorously account for lag and alignment in moving averages.
  \item We prove a core identity showing that:
  \[
  \bar{P}^a(x) - \bar{P}^{2a}(x) = \frac{a}{2} \cdot \frac{d}{dx} \tilde{\tilde{P}}^a(x - a),
  \]
  thereby interpreting MACD as a derivative of a doubly phase-corrected smoothing operator.
  \item We develop a recursive expansion for nested averages, culminating in an identity:
  \[
  \bar{P}^{a}(x) - \bar{P}^{a+b}(x)
= \frac{b}{2} \frac{d}{dx} \sum_{i=1}^{n} \frac{2i}{n(n+1)} \dcavg{P}{b}(x - ib)\]
  where \( a = nb \), revealing a deeper structure in multi-scale averaging.
  \item We interpret this structure as a band-pass filter applied to a smooth approximation of the input signal, thus connecting MACD to classical tools from harmonic analysis and signal processing.
\end{itemize}

This unified perspective clarifies how MACD balances smoothing and sensitivity to local change, explains its effectiveness in trend detection, and offers a foundation for principled generalizations.

\section{Averaging Operators and Recursive Composition}

We begin by defining a standard moving average operator over a time window \( T \):

\begin{definition}[Right-endpoint moving average]
Let \( f: \mathbb{R} \to \mathbb{R} \) be a locally integrable function. The \textit{right-endpoint average} of \( f \) over an interval of length \( T > 0 \) is defined as:
\[
\bar{f}^T(x) := (R_T f)(x) :=  \frac{1}{T} \int_{x - T}^x f(t) \, dt.
\]

\end{definition}

This operator acts as a smoothing kernel over the last \( T \) units of time up to position \( x \).

\begin{theorem}[Boundedness of MACD Operator in $L^p$]
Let $P \in L^p(\mathbb{R})$ for $1 \le p \le \infty$. 
Define the MACD operator:
\[
M_a[P](x) = \bar{P}^a(x) - \bar{P}^{2a}(x),
\]
where $\bar{P}^a = P * K_a$ and $K_a = \frac1a \chi_{[0,a]}$.
Then $M_a$ is bounded on $L^p$ and satisfies
\[
\| M_a[P] \|_{L^p} \le 2 \| P \|_{L^p}.
\]
\end{theorem}
\begin{proof}
Since $\bar{P}^a = P * K_a$ and $\|K_a\|_{L^1} = 1$, Young's inequality gives
\[
\|\bar{P}^a\|_{L^p} \le \|P\|_{L^p}, \qquad
\|\bar{P}^{2a}\|_{L^p} \le \|P\|_{L^p}.
\]
By the triangle inequality:
\[
\|M_a[P]\|_{L^p} \le \|\bar{P}^a\|_{L^p} + \|\bar{P}^{2a}\|_{L^p}
\le 2 \|P\|_{L^p}.
\]
\end{proof}

\section{Recursive Structure of Averaging Operators}

Let us suppose we want to compute \( \bar{P}^T(x) \) using two shorter intervals \( t_1 \) and \( t_2 \), such that:
\[
T = t_1 + t_2.
\]

\begin{theorem}[Recursive decomposition of right-endpoint moving average]
Let \( T = t_1 + t_2 \). Then the right-endpoint average satisfies:
\[
\bar{P}^T(x) = \frac{t_1}{T} \bar{P}^{t_1}(x) + \frac{t_2}{T} \bar{P}^{t_2}(x - t_1),
\]
where
\[
\bar{P}^t(x) := \frac{1}{t} \int_{x - t}^{x} P(s) \, ds.
\]
\end{theorem}

\begin{proof}
By definition,
\[
\bar{P}^T(x) = \frac{1}{T} \int_{x - T}^{x} P(s) \, ds.
\]

Split the integral at \( x - t_1 \) (note that \( t_2 = T - t_1 \)):
\[
\int_{x - T}^{x} P(s) \, ds
= \int_{x - T}^{x - t_2} P(s) \, ds + \int_{x - t_2}^{x} P(s) \, ds.
\]

Observe that:
\[
\int_{x - T}^{x - t_2} P(s) \, ds = t_2 \cdot \bar{P}^{t_2}(x - t_1),
\quad
\int_{x - t_2}^{x} P(s) \, ds = t_1 \cdot \bar{P}^{t_1}(x).
\]

Therefore:
\[
\bar{P}^T(x) = \frac{1}{T} \left( t_1 \bar{P}^{t_1}(x) + t_2 \bar{P}^{t_2}(x - t_1) \right),
\]
as required.
\end{proof}

\begin{corollary}[Monotonicity in Averaging Window]
Let \( a, b > 0 \) with \( b > a \), and suppose:
\[
\bar{P}^a(x) > \bar{P}^b(x),
\]
where \( \bar{P}^t(x) := \frac{1}{t} \int_{x - t}^{x} P(s) \, ds \). Then it follows that:
\[
\bar{P}^a(x) > \bar{P}^{b - a}(x - a).
\]
\end{corollary}

\begin{proof}
From Theorem 1, we know that:
\[
\bar{P}^{b}(x) = \frac{a}{b} \bar{P}^{a}(x) + \frac{b - a}{b} \bar{P}^{b - a}(x - a).
\]

Assume, for contradiction, that \( \bar{P}^a(x) \leq \bar{P}^{b - a}(x - a) \). Then:
\[
\bar{P}^b(x) \geq \bar{P}^a(x),
\]
because a convex combination of \( \bar{P}^a(x) \) and a larger (or equal) value yields a value at least as large as \( \bar{P}^a(x) \).

This contradicts the assumption that \( \bar{P}^a(x) > \bar{P}^b(x) \). Therefore, it must be that:
\[
\bar{P}^a(x) > \bar{P}^{b - a}(x - a).
\]
\end{proof}

\begin{corollary}[Equality Case]
Let \( a, b > 0 \) with \( b > a \), and suppose:
\[
\bar{P}^a(x) = \bar{P}^b(x),
\]
where \( \bar{P}^t(x) := \frac{1}{t} \int_{x - t}^{x} P(s) \, ds \). Then it follows that:
\[
\bar{P}^a(x) = \bar{P}^{b - a}(x - a).
\]
\end{corollary}

\begin{remark}[Interpretation of Moving Average Differences]

Let \( a, b > 0 \), and suppose we compare \( \bar{P}^{a}(x) \) and \( \bar{P}^{b + a}(x) \), both anchored at the same point \( x \). Using corollary, we can deduce the following:

\begin{itemize}
    \item If \( \bar{P}^{a}(x) > \bar{P}^{b + a}(x) \), then it must be that \( \bar{P}^{a}(x) > \bar{P}^{b}(x - a) \), implying that \( P(x) \) is \textbf{locally increasing}.
    
    \item If \( \bar{P}^{a}(x) < \bar{P}^{b + a}(x) \), then \( \bar{P}^{a}(x) < \bar{P}^{b}(x - a) \), indicating \( P(x) \) is \textbf{locally decreasing}.
    
    \item If \( \bar{P}^{a}(x) = \bar{P}^{b + a}(x) \), then \( \bar{P}^{a}(x) = \bar{P}^{b}(x - a) \), suggesting \( P(x) \) is \textbf{locally linear} (or symmetric) across the interval \( [x - a - b, x] \).
\end{itemize}

This formulation provides a precise way to assess local monotonicity by comparing moving averages over nested intervals that end at the same point. This interpretation aligns naturally with the idea that comparing moving averages of different spans can serve as a proxy for estimating the \textbf{direction and intensity of local trends} — a core principle in momentum-based indicators like MACD.
\end{remark}

\section{MACD as a Derivative of Smoothed Signals}
\subsection{Difference Identity and Recursive Relation}

We begin by establishing a precise identity that relates moving averages of different lengths:

\begin{lemma}[Recursive Difference Identity]
Let \( a, b > 0 \). Then the following identity holds:
\[
\bar{P}^a(x) - \bar{P}^{a + b}(x)
= \frac{b}{a + b} \left( \bar{P}^a(x) - \bar{P}^b(x - a) \right),
\]
where \( \bar{P}^t(x) := \frac{1}{t} \int_{x - t}^{x} P(s) \, ds \).
\end{lemma}

\begin{proof}
From the recursive averaging identity (Theorem 1), we have:
\[
\bar{P}^{a + b}(x) = \frac{a}{a + b} \bar{P}^a(x) + \frac{b}{a + b} \bar{P}^b(x - a).
\]
Rewriting:
\[
\bar{P}^a(x) - \bar{P}^{a + b}(x)
= \bar{P}^a(x) - \left( \frac{a}{a + b} \bar{P}^a(x) + \frac{b}{a + b} \bar{P}^b(x - a) \right),
\]
\[
= \left( 1 - \frac{a}{a + b} \right) \bar{P}^a(x) - \frac{b}{a + b} \bar{P}^b(x - a)
= \frac{b}{a + b} \left( \bar{P}^a(x) - \bar{P}^b(x - a) \right).
\]
\end{proof}

\begin{remark}
This identity refines the inequality-based corollary from earlier: when \( \bar{P}^a(x) > \bar{P}^b(x - a) \), the left-hand side is positive, quantifying the difference. It also sets up an exact algebraic path to interpreting MACD as a smoothed derivative.
\end{remark}

\subsection{Definition: Double Averaging}

\begin{definition}[Double Averaging Operator]
Let \( \bar{P}^a(x) \) denote the right-endpoint average of \( P(x) \). The \textbf{double average} of \( P \) over window \( a \) is defined as:
\[
\bar{P}^{a,a}(x) := R_a \left[ \bar{P}^a(x) \right]
= \frac{1}{a} \int_{x - a}^{x} \bar{P}^a(s) \, ds.
\]
\end{definition}

\subsection{Main Result: MACD as a Derivative}

\begin{theorem}[MACD as Derivative of Double Average]
Let \( P(x) \in C^0(I) \), and consider the difference:
\[
\bar{P}^a(x) - \bar{P}^{2a}(x).
\]
Then the following identity holds:
\[
\bar{P}^a(x) - \bar{P}^{2a}(x)
= \frac{a}{2} \frac{d}{dx} \bar{P}^{a,a}(x).
\]
\end{theorem}

\begin{proof}
We begin by using the recursive decomposition:
\[
\bar{P}^{2a}(x) = \frac{1}{2} \bar{P}^{a}(x) + \frac{1}{2} \bar{P}^{a}(x - a).
\]
Therefore:
\[
\bar{P}^a(x) - \bar{P}^{2a}(x)
= \frac{1}{2} \left( \bar{P}^a(x) - \bar{P}^a(x - a) \right).
\]

Now observe:
\[
\bar{P}^{a,a}(x) = \frac{1}{a} \int_{x - a}^{x} \bar{P}^a(s) \, ds,
\]
and differentiating:
\[
\frac{d}{dx} \bar{P}^{a,a}(x)
= \frac{1}{a} \left( \bar{P}^a(x) - \bar{P}^a(x - a) \right).
\]

Thus:
\[
\bar{P}^a(x) - \bar{P}^{2a}(x) = \frac{a}{2} \frac{d}{dx} \bar{P}^{a,a}(x),
\]
as claimed.
\end{proof}

\begin{remark}
This theorem gives a precise analytical interpretation of MACD as the derivative of a smoothed signal — specifically, a doubly-averaged version of \( P(x) \). In contrast to raw derivatives, this formulation emphasizes stability and denoising, a key reason for MACD’s effectiveness in signal processing and trend detection.
\end{remark}

\begin{remark}[Regularity of Averaging Operators]
The identity
\[
\bar{P}^a(x) - \bar{P}^{2a}(x) = \frac{a}{2} \frac{d}{dx} \bar{P}^{a,a}(x)
\]
does not require \( P(x) \) to be differentiable. It holds for any locally integrable function \( P \in L^1_{\mathrm{loc}} \), since the moving average \( \bar{P}^a \) is defined as a convolution with a compactly supported, integrable kernel.

Thus, even when \( P(x) \) is merely continuous or piecewise smooth, the right-hand side derivative in the MACD identity is well-defined.
\end{remark}

\begin{theorem}[Regularity Gain of Averaging]
Let \( f \in C^n(\mathbb{R}) \). Then the right-endpoint averaging operator \( \bar{f}^a(x) \), defined as
\[
\bar{f}^a(x) := \frac{1}{a} \int_{x - a}^{x} f(s) \, ds,
\]
produces a function in \( C^{n+1}(\mathbb{R}) \). That is:
\[
f \in C^n \quad \Rightarrow \quad \bar{f}^a \in C^{n+1}.
\]
\end{theorem}

\begin{proof}
The averaging operator is equivalent to convolution with the kernel:
\[
K_a(s) := \frac{1}{a} \chi_{[0, a]}(s),
\]
which is compactly supported, bounded, and piecewise continuous. Convolution with such a kernel raises regularity by one:
\[
\bar{f}^a(x) = (f * K_a)(x) \in C^{n+1}.
\]
\end{proof}

\section{MACD as Derivative of Double Phase-Corrected Average}

We refine our understanding of MACD by introducing \textbf{phase-corrected averaging}, and showing that the MACD operator is precisely a \textbf{delayed derivative of a double phase-corrected average}.

\begin{definition}[Phase-Corrected Averaging]
Given a locally integrable function \( P(x) \), we define the \textbf{phase-corrected moving average} of window length \( a > 0 \) as:
\[
\tilde{f}^a(x) := (C_a f)(x) := \frac{1}{a} \int_{x - \frac{a}{2}}^{x + \frac{a}{2}} f(s) \, ds = \Bar{f}^a\left(x+\frac{a}{2}\right).
\]
This represents a symmetric (centered) average around \( x \), unlike the causal right-endpoint average \( \bar{f}^a(x) \).
\end{definition}

\begin{definition}[Double Phase-Corrected Averaging]
We define the \textbf{double phase-corrected average} as the phase-corrected average of the phase-corrected signal:
\[
\tilde{P}^{a,a}(x) := \dcavg{P}{a}(x) := \frac{1}{a} \int_{x - \frac{a}{2}}^{x + \frac{a}{2}} \tilde{P}^a(s) \, ds.
\]
Equivalently, \( \tilde{P}^{a,a}(x) =C_a\left[ C_a P \right](x) \), i.e., a composition of centered smoothing operators.
\end{definition}

\begin{corollary}[MACD as Delayed Derivative of Double Phase-Corrected Average]
Let \( P(x) \) be a locally integrable function. Then the MACD difference between short- and long-term averages satisfies the identity:
\[
\bar{P}^a(x) - \bar{P}^{2a}(x)
= \frac{a}{2} \cdot \frac{d}{dx} \tilde{P}^{a,a}(x - a),
\]
where \( \tilde{P}^{a,a}(x) \) is the double phase-corrected average.
\end{corollary}

\begin{proof}
From earlier results, we know that:
\begin{align*}
    \bar{P}^a(x) - \bar{P}^{2a}(x) &= \frac{a}{2} \frac{d}{dx} \bar{P}^{a,a}(x)\\
    &=\frac{a}{2} \frac{d}{dx} C_a \bar{P}^{a}\left(x-\frac{a}{2}\right)\\
    &=\frac{a}{2} \frac{d}{dx} C_a^2 P\left(x-a\right)\\
    &=\frac{a}{2}\frac{d}{dx} \tilde P^{a,a}\left(x-a\right)
\end{align*}
\end{proof}

\begin{remark}
This result reveals the exact structure underlying MACD: it is a **derivative of a doubly smoothed, phase-centered signal**, delayed to account for the implicit lag introduced by successive right-endpoint averages. This perspective eliminates heuristic lag and asymmetry, and shows MACD as a well-defined, linear differential operator on smoothed signals.
\end{remark}

\section{Recursive Decomposition and Delay Expansion}

We now develop a recursive expression for the MACD-like difference between a short-term and a long-term average when the long-term window is an integer multiple of the short one. Our derivation of MACD as a weighted sum of smoothed derivatives relates conceptually to triangular convolution kernels in harmonic analysis \cite{folland1999, rudin1991}.

\subsection{7.1 Difference Identity for Nested Averages}

\begin{theorem}[Recursive Derivative Expansion of Averaging Difference]
Let \( b > 0 \), and let \( n \in \mathbb{N} \). Define \( a := nb \). Then the difference between short- and long-term moving averages satisfies the identity:
\[
\bar{P}^{a}(x) - \bar{P}^{a+b}(x)
= \frac{b}{2} \frac{d}{dx} \sum_{i=1}^{n} \frac{2i}{n(n+1)} \dcavg{P}{b}(x - ib),
\]
where each \( \bar{P}^b \) is the average over window \( b \) evaluated at a delayed time.
\end{theorem}

\begin{proof}\begin{align*}
    \Bar{P}^a(x) - \Bar{P}^{a+b}(x) =& \frac{b}{(n+1)b} \left(\Bar{P}^{nb}(x) - \Bar{P}^{b}(x - nb)\right) \\
    =& \frac{1}{n+1} \left( \frac{n-1}{n} \Bar{P}^{(n-1)b}(x) + \frac{1}{n} \Bar{P}^{b}(x - (n-1)b) - \Bar{P}^{b}(x - nb) \right)\\
    =& \frac{1}{n+1} \left( \frac{n-1}{n} \Bar{P}^{(n-1)b}(x) - \frac{n-1}{n} \Bar{P}^{b}(x - (n-1)b) + \Bar{P}^{b}(x - (n-1)b) - \Bar{P}^{b}(x - nb) \right)\\
     =& \frac{1}{(n+1)} \cdot \frac{n-1}{n} \left( \Bar{P}^{(n-1)b}(x) - \Bar{P}^b(x - (n-1)b) \right) + \frac{b}{(n+1)} \frac{d}{dx} \davg{P}{b}(x - nb + b) \\
    =& \frac{n-1}{(n+1)n} \left( \Bar{P}^{(n-1)b}(x) - \Bar{P}^b(x - (n-1)b) \right) + \frac{b}{(n+1)} \frac{d}{dx} \davg{P}{b}(x - (n-1)b) \\
    =& \frac{1}{(n+1)}\left[b \frac{d}{dx} \davg{P}{b}(x - (n-1)b)+\frac{n-1}{n} \left( \Bar{P}^{(n-1)b}(x) - \Bar{P}^b(x - (n-1)b) \right)\right]\\
    =&\frac{1}{(n+1)}\left[b \frac{d}{dx} \davg{P}{b}(x - (n-1)b)\right.\\
    &+\frac{n-1}{n}\left(b \frac{d}{dx} \davg{P}{b}(x - (n-2)b)+\frac{n-2}{n-1}\left(  \Bar{P}^{(n-1)b}(x) 
     - \Bar{P}^b(x - (n-1)b) \right)\right ]\\
    =&\hdots \\
    =& \frac{b}{(n+1)n}\frac{d}{dx} \davg{P}{b}(x)+\frac{2b}{(n+1)n}\frac{d}{dx} \davg{P}{b}(x-b)+\hdots+\frac{nb}{(n+1)n}\frac{d}{dx} \davg{P}{b}(x-(n-1)b)\\
    =&\sum_{i=1}^n\frac{ib}{(n+1)n}\frac{d}{dx}\davg{P}{b}(x-(i-1)b)\\
\end{align*}

Hence,
\begin{align*}
    \Bar{P}^a(x) - \Bar{P}^{a+b}(x) =\frac{b}{2}\frac{d}{dx}&\sum_{i=1}^n\frac{2i}{(n+1)n}\davg{P}{b}(x-(i-1)b)\\
    = \frac{b}{2}\frac{d}{dx}&\sum_{i=1}^n\frac{2i}{(n+1)n}\dcavg{P}{b}(x-ib)\\
\end{align*}

\end{proof}

\subsection{Interpretation and Application}

This final identity reveals that MACD, when generalized over recursive average windows, becomes a \textbf{differentiated convolution sum} over time-delayed smoothed signals. The weights \( \frac{i}{n} \) represent a \textbf{linearly increasing emphasis on more recent history}, while the derivative extracts the overall directional trend.

This aligns perfectly with the intuitive and empirical function of MACD: it is a \textbf{band-pass operator}, sensitive to changes in trend over intermediate time scales, and suppressing both short-term noise and long-term drift.

\begin{remark}
The form:
\[
\frac{b}{2} \frac{d}{dx} \sum_{i=1}^{n} \frac{2i}{n(n+1)} \dcavg{P}{b}(x - ib)\]
is a concrete realization of MACD as a differentiation of convolution with a \textbf{triangular kernel}.
\end{remark}

\section*{Conclusion}

In this paper, we have reformulated the MACD indicator as a rigorously defined linear operator rooted in classical analysis. Through recursive decomposition, phase correction, and double averaging, we demonstrated that MACD can be viewed as a delayed derivative applied to a doubly smoothed signal — a structure that reflects both stability and responsiveness.

We unify empirical technical analysis tools with rigorous operator theory \cite{murphy1999, oppenheim2009}, offering an analytical foundation for future research into adaptive filters and financial signal processing. This framework not only justifies empirical practices but also opens a path to principled generalizations of momentum-based indicators and signal-processing strategies in financial and dynamical systems.

Future work may explore Fourier-domain representations, optimal weighting schemes, and extensions to multi-dimensional signals and adaptive filtering.

\bibliographystyle{plain}
\bibliography{refs}

\end{document}